\documentclass{article}
\usepackage{graphicx}
\usepackage{multirow}

\usepackage{color}
\usepackage{framed}

\newcommand{\aop}{Ann. Phys.~}
\newcommand{\cmp}{Comm. Math. Phys.~}

\newcommand{\jmp}{J. Math. Phys.~}
\newcommand{\jpa}{J. Phys. A~}

\newcommand{\njp}{New. J. Phys.~}

\newcommand{\prl}{Phys. Rev. Lett.~}
\newcommand{\pra}{Phys. Rev. A~}
\newcommand{\pre}{Phys. Rev. E~}
\newcommand{\pla}{Phys. Lett. A~}

\usepackage[T1]{fontenc}
\usepackage[sc]{mathpazo}
\usepackage{amsmath}
\usepackage{amssymb}
\usepackage{enumerate}
\usepackage{amsthm}
\usepackage{amsfonts,mathrsfs}
\usepackage[vcentermath]{youngtab}

\usepackage{hyperref}
\hypersetup{pdfpagemode=UseNone}

\newcommand{\tinyspace}{\mspace{1mu}}

\newcommand{\op}[1]{\operatorname{#1}}

\newcommand{\abs}[1]{\left\lvert\tinyspace #1 \tinyspace\right\rvert}

\renewcommand{\det}{\operatorname{det}}

\newcommand{\setft}[1]{\mathrm{#1}}
\newcommand{\lin}[1]{\setft{L}\left(#1\right)}

\renewcommand{\vec}{\op{vec}}
\newcommand{\im}{\op{im}}

\def\complex{\mathbb{C}}
\def\real{\mathbb{R}}

\def\I{\mathbb{1}}

\def \dif {\mathrm{d}}
\def \diag {\mathrm{diag}}

\def \re {\mathrm{Re}}
\def \im {\mathrm{Im}}

\newenvironment{mylist}[1]{\begin{list}{}{
    \setlength{\leftmargin}{#1}
    \setlength{\rightmargin}{0mm}
    \setlength{\labelsep}{2mm}
    \setlength{\labelwidth}{8mm}
    \setlength{\itemsep}{0mm}}}
    {\end{list}}

\def\ot{\otimes}

\newcommand{\out}[2]{| #1\rangle\langle #2 |}

\newcommand{\Innerm}[3]{\left\langle #1 \left| #2 \right| #3 \right\rangle}

\newcommand{\pa}[1]{(#1)}
\newcommand{\Pa}[1]{\left(#1\right)}

\newcommand{\Br}[1]{\left[#1\right]}
\newcommand{\set}[1]{\{#1\}}
\newcommand{\Set}[1]{\left\{#1\right\}}

\newcommand{\bra}[1]{\langle#1|}

\newcommand{\ket}[1]{|#1\rangle}

\DeclareMathOperator{\trace}{Tr}
\newcommand{\ptr}[2]{\trace_{#1}\pa{#2}}
\newcommand{\Ptr}[2]{\trace_{#1}\Pa{#2}}

\newcommand{\Tr}[1]{\Ptr{}{#1}}

\def\cE{\mathcal{E}}
\def\cH{\mathcal{H}}\def\cI{\mathcal{I}}

\def\cX{\mathcal{X}}\def\cY{\mathcal{Y}}

\def\rS{\mathrm{S}}

\def\sC{\mathscr{C}}

\newtheorem{thrm}{Theorem}[section]

\newtheorem{prop}[thrm]{Proposition}

\theoremstyle{definition}
\newtheorem{definition}[thrm]{Definition}
\newtheorem{remark}[thrm]{Remark}

\numberwithin{equation}{section}

\newcounter{questionnumber}

\begin{document}

\title{\Large Average coherence and its typicality for random mixed quantum states}

\author{Lin Zhang\footnote{E-mail: linyz@zju.edu.cn; godyalin@163.com}\\
  {\it\small Institute of Mathematics, Hangzhou Dianzi University, Hangzhou 310018, PR~China}}
\date{}
\maketitle
\maketitle \mbox{}\hrule\mbox\\
\begin{abstract}
Wishart ensemble is a useful and important random matrix model used
in diverse fields. By realizing induced random mixed quantum states
as Wishart ensemble with the fixed-trace one, using matrix integral
technique we give a fast track to the average coherence for random
mixed quantum states induced via partial-tracing of the
Haar-distributed bipartite pure states. As a direct consequence of
this result, we get a compact formula of the average subentropy of
random mixed states. These obtained compact formulae extend our
previous work.

\end{abstract}
\maketitle \mbox{}\hrule\mbox

\section{Introduction}

Quantum coherence, due to the superposition rule, is an important
ingredient in quantum information processing and plays a pivotal
role in diverse fields such as quantum thermodynamics
\cite{Brandao2015,Rodr2013} and quantum biology
\cite{Levi2014,Lloyd2011,Plenio2008}. Quantum coherence is the basis
of single particle interferometry, and it gives coherence the status
of a resource and makes necessary to develop a solid framework
allowing to asses and quantify this property \cite{Baumgratz2014}.
Quantum entanglement, due to the tensor product structure of
composite quantum systems, is also another fundamental feature of
quantum mechanics. It is also necessary resource in many quantum
information processing tasks such as superdense coding, quantum
teleportation etc. \cite{watrous2016preprint}. Recently researchers
contributes much effort to connect quantum coherence with
entanglement and quantum discord, a kind of quantum correlation
containing entanglement as a proper subset. Streltsov \emph{et. al}
\cite{Streltsov2015} have connected quantum coherence with
entanglement, and have shown that any degree of coherence with
respect to some reference basis can be converted to entanglement via
incoherent operations. Ma \emph{et. al} \cite{Ma2016} have proven
that the creation of quantum discord with multipartite incoherent
operations is bounded by the amount of quantum coherence consumed in
its subsystems during the process.

In the last years, many efforts have been made towards the research
of quantum correlations of random quantum states
\cite{Hall1998,Wootters1990}. In quantum information theory, many
quantities such as the quantum entanglement and the diagonal entropy
of a density matrix \cite{Giraud2016} has been proved very useful.
In particular, typicality of some quantity can reduce computational
complexity of it \cite{Hayden2006}. For example, the typicality
entanglement of pure bipartite states sampled randomly according to
the uniform Haar measure provides an explanation to the equal a
priori postulate of the statistical physics \cite{Goldstein2006}. To
the knowledge of the author, the distribution of entanglement among
two subsystem of a large quantum system has been a subject of
interest among physicists and mathematicians for a long time, and
many interesting results have been obtained, but however, similar
consideration for quantum coherence is still missing. In the present
work, we are concerned with the statistical behavior of coherence of
a subsystem of a large quantum system. Specifically, although the
authors \cite{Zhang2015} make an attempt to calculating the average
coherence of \emph{induced} random mixed state ensemble
\cite{Karol2001} by brute force, the calculation is very complicated
and also tedious. This motivates me to find a more compact approach
to the typicality of quantum coherence. Luckily, we find a simple
approach to get a compact formula for the average coherence quickly.
Although the topic of the present paper was already investigated and
some results was obtained \cite{Zhang2015}, the method used in the
paper is new and very different from that of \cite{Zhang2015}. The
authors of \cite{Zhang2015} by calculating the average subentropy
firstly and then deriving the average coherence by using the
obtained formula for the average subentropy. In obtaining the final
compact forms for the average subentropy and the average coherence,
they have shown some ingenious combinatorial identities (we can see
this from the very recent published version). But, however, what is
the difference is that we calculate directly the average coherence
(more simpler than the method in \cite{Zhang2015}) and get a compact
form for the average subentropy as a by-product. Based on this
elegant formula for the average subentropy, we get the fact that as
the dimension of the system to be considered increases, the average
subentropy of random mixed states approaches to the maximum value of
the subentropy which is attained for the maximally mixed state.

Let us fix some notations before proceeding. For a given density
matrix $\rho$ (i.e. nonnegative square matrix of or $m$ with unit
trace), its von Neumann entropy is defined as $S(\rho):=
-\Tr{\rho\ln\rho}$, where $\ln\rho$ is in the sense of the
functional calculus of $\rho$. In fact,
$\rS(\rho)=-\sum_j\lambda_j(\rho)\ln\lambda_j(\rho)$, where
$\lambda_j(\rho)$ stands for the eigenvalues of $\rho$. Quantum
relative entropy of coherence (in short quantum coherence in the
present paper) in a state $\rho$ is given by \cite{Baumgratz2014}:
$\sC(\rho) = S(\rho_\diag) - S(\rho)$. In this paper, we will
calculate exactly the average coherence $\overline{\sC}$ for random
mixed quantum states. Then the typicality of coherence is obtained
immediately.

The paper is organized as follows. In Sect.~\ref{sect:matrix-delta},
We recall the Dirac delta function and its extension to matrix delta
function. Then, We derive the distribution of diagonal part of
Wishart ensemble, as the marginal distribution of matrix elements of
Wishart ensemble in Sect.~\ref{sect:wishart}, and by realizing the
induced random mixed quantum states as Wishart ensemble with
fixed-trace one, we obtain the distribution of diagonal part of
induced random mixed quantum states. In this section, we also
calculate the average entropy of diagonal part of random mixed
quantum states. We present our main results (i.e.,
Theorems~\ref{th:ave-coh}, \ref{th:ave-suben}, and \ref{th:typ-coh})
in Sect.~\ref{sect:main-results}. In Sect.~\ref{sect:mixing}, we
discuss the average coherence of the mixing of random mixed quantum
states. Finally, we give concluding remarks in
Sect.~\ref{sect:conclusion}.

\section{Matrix delta function}\label{sect:matrix-delta}

Although the matrix delta function has already been used in the
literatures, but there is no formal and rigorous treatment, to my
best knowledge. For reader's convenience, we will give a complete
detail along this line. We recall that Dirac delta function
$\delta(x)$ \cite{Hoskins2009} is defined by
\begin{eqnarray}
\delta(x) = \begin{cases} +\infty,&\text{if } x=0;\\0, &\text{if }
x\neq 0.
\end{cases}
\end{eqnarray}
The Fourier integral representation of Dirac delta function
\begin{eqnarray}
\delta(x) = \frac1{2\pi}\int^{+\infty}_{-\infty}
e^{\mathrm{i}tx}\dif t\quad (x\in\real)
\end{eqnarray}
cab be extended to the matrix case.
\begin{definition}[The matrix delta
function] (i) For an $m\times n$ complex matrix $Z=[z_{ij}]$, the
matrix delta function $\delta(Z)$ is defined as
\begin{eqnarray}
\delta(Z):=\prod^m_{i=1}\prod^n_{j=1}\delta\Pa{\re(z_{ij})}\delta\Pa{\im(z_{ij})}.
\end{eqnarray}
(ii) For an $m\times m$ Hermitian complex matrix $X=[x_{ij}]$, the
matrix delta function $\delta(X)$ is defined as
\begin{eqnarray}
\delta(X):=\prod_j \delta(x_{jj})\prod_{i<
j}\delta\Pa{\re(x_{ij})}\delta\Pa{\im(x_{ij})}.
\end{eqnarray}
\end{definition}
From the above definition, we see that the matrix delta function of
a complex matrix is equal to the product of one-dimensional delta
functions over the independent real and imaginary parts of this
complex matrix. The following proposition is very important in this
paper.

\begin{prop}[The Fourier integral representation of the matrix delta function]
For an $m\times m$ Hermitian complex matrix $X$, we have
\begin{eqnarray}\label{eq:matrix-delta-function}
\delta(X) = \frac1{2^m\pi^{m^2}}\int e^{\mathrm{i}\Tr{TX}}[\dif T],
\end{eqnarray}
where $T=[t_{ij}]$ is also an $m\times m$ Hermitian complex matrix,
and $[\dif T]:=\prod_j\dif t_{jj}\prod_{i<j}\dif \re(t_{ij})\dif
\im(t_{ij})$.
\end{prop}

\begin{proof}
Indeed, we know that
\begin{eqnarray*}
\Tr{TX} &=& \sum_{j=1}^m t_{jj}x_{jj} + \sum_{i\neq j}\Pa{\bar
t_{ij}x_{ij}}= \sum_{j=1}^m \re(t_{jj})\re(x_{jj}) +
\sum_{1\leqslant i<
j\leqslant m}\Pa{\bar t_{ij}x_{ij} + t_{ij}\bar x_{ij}}\\
&=&\sum_{j=1}^m t_{jj}x_{jj} + \sum_{1\leqslant i< j\leqslant
m}2\Pa{\re(t_{ij})\re(x_{ij}) + \im(t_{ij})\im(x_{ij})},
\end{eqnarray*}
implying that
\begin{eqnarray*}
&&\int e^{\mathrm{i}\Tr{TX}}[\dif T] \\
&&= \prod_{j=1}^m \int \exp\Pa{\mathrm{i}t_{jj}x_{jj}}\dif
t_{jj}\prod_{1\leqslant i< j\leqslant m}\int
\exp\Pa{\mathrm{i}\re(t_{ij})\Pa{2\re(x_{ij})}}\dif
\re(t_{ij})\int\exp\Pa{\mathrm{i}\im(t_{ij}) \Pa{2\im(x_{ij})}}\dif
\im(t_{ij})\\
&&= \prod^m_{j=1}2\pi\delta\Pa{x_{jj}}\times \prod_{1\leqslant
i<j\leqslant m}2\pi\delta\Pa{2\re(x_{ij})}2\pi\delta\Pa{2\im(x_{ij})}\\
&&=\prod^m_{j=1}2\pi\delta\Pa{x_{jj}}\times \prod_{1\leqslant
i<j\leqslant m}\pi\delta\Pa{\re(x_{ij})}\pi\delta\Pa{\im(x_{ij})}\\
&&=\Pa{2\pi}^m\Pa{\pi^2}^{\binom{m}{2}}\prod_j\delta(x_{jj})\prod_{i< j}\delta\Pa{\re(x_{ij})}\delta\Pa{\im(x_{ij})}\\
&&= 2^m\pi^{m^2}\delta(X).
\end{eqnarray*}
Therefore we get the desired identity.
\end{proof}

\begin{remark}
Indeed, since
\begin{eqnarray*}
\Tr{T^{\mathrm{off}}X^{\mathrm{off}}} =
\sum_{i<j}2(\re(t_{ij})\re(x_{ij}) + \im(t_{ij})\im(x_{ij}))
\end{eqnarray*}
and
\begin{eqnarray*}
[\dif T^{\mathrm{off}}] = \prod_{i<j}\dif\re(t_{ij})\dif
\im(t_{ij}),
\end{eqnarray*}
it follows that
\begin{eqnarray*}
&&\int[\dif T^{\mathrm{off}}]\exp\Pa{\mathrm{i}\Tr{T^{\mathrm{off}}X^{\mathrm{off}}}}\\
&&=
\prod_{i<j}\int\dif\re(t_{ij})\exp\Pa{\mathrm{i}\re(t_{ij})(2\re(x_{ij}))}\int\dif\im(t_{ij})\exp\Pa{\mathrm{i}\im(t_{ij})(2\im(x_{ij}))}\\
&&=\prod_{i<j}2\pi\delta(2\re(x_{ij}))\cdot2\pi\delta(2\im(x_{ij}))
=
\prod_{i<j}\pi\delta(\re(x_{ij}))\pi\delta(\im(x_{ij}))\\
&&=\pi^{2\binom{m}{2}}\prod_{i<j}\delta(\re(x_{ij}))\delta(\im(x_{ij}))
= \pi^{m(m-1)}\delta(X^{\mathrm{off}}).
\end{eqnarray*}
From the above discussion, we see that
\eqref{eq:matrix-delta-function} can be separated into two
identities below:
\begin{eqnarray}
\delta(X^\diag) &=& \frac1{(2\pi)^m}\int[\dif
T^\diag]e^{\mathrm{i}\Tr{T^\diag X^\diag}},\\
\delta(X^\mathrm{off}) &=& \frac1{\pi^{m(m-1)}}\int[\dif
T^\mathrm{off}]e^{\mathrm{i}\Tr{T^\mathrm{off}X^\mathrm{off}}}.
\end{eqnarray}
\end{remark}

\section{Wishart ensemble}\label{sect:wishart}

We use the notation $x\sim N(\mu,\sigma^2)$ to indicate a Gaussian
random real variable $x$ with mean $\mu$ and variance $\sigma^2$.
Let $Z$ denote an $m\times n(m\leqslant n)$ complex random matrix
\cite{Hiai2000,Mehta}. These elements are independent complex random
variables subject to
$N_\complex(0,1)=N(0,\frac12)+\mathrm{i}N(0,\frac12)$ with Gaussian
densities:
\begin{eqnarray}
\frac1\pi \exp\Pa{-\abs{z_{ij}}^2},
\end{eqnarray}
where $\re(z_{ij}),\im(z_{ij})$ are i.i.d. Gaussian random real
variables with mean 0 and variance $\frac12$.

\begin{definition}[Wishart matrices, \cite{James1964}]
With $m\times n$ random matrices $Z$ specified as above, define
complex Wishart ensemble as consisting of matrices $W=ZZ^\dagger$.
The matrices $W=ZZ^\dagger$ are referred to as (uncorrelated)
\emph{Wishart matrices}.
\end{definition}
As chosen previously, $m\leqslant n$ for definiteness. The
probability distribution followed by $Z$ is
\begin{eqnarray}
\varphi(Z)\propto \exp\Pa{-\Tr{ZZ^\dagger}}.
\end{eqnarray}
Indeed, let $Z=[z_{ij}]$ be a complex random matrix, where
$z_{ij}=\re(z_{ij})+\sqrt{-1}\im(z_{ij})$ with
$\re(z_{ij}),\im(z_{ij})\sim N(0,\frac12)$. The probability
distribution of $Z$ is just the joint distribution of all matrix
elements $z_{ij}$ of $Z$. Thus
\begin{eqnarray*}
\varphi(Z)=\prod^m_{i=1}\prod^n_{j=1}\mathrm{Pr}(z_{ij})=\prod^m_{i=1}\prod^n_{j=1}\frac1\pi\exp\Pa{-\abs{z_{ij}}^2}=\frac1{\pi^{mn}}\exp\Pa{-\sum^m_{i=1}\sum^n_{j=1}\abs{z_{ij}}^2}.
\end{eqnarray*}
Thus, we have
\begin{eqnarray}
\varphi(Z) = \frac1{\pi^{mn}}\exp\Pa{-\Tr{ZZ^\dagger}}.
\end{eqnarray}
Then the distribution of Wishart matrices $W$ is given by
\begin{eqnarray}\label{eq:dist-for-W}
\mathbf{Q}(W) = \int[\dif Z]\delta(W-ZZ^\dagger)\varphi(Z),
\end{eqnarray}
where $[\dif Z]=\prod^m_{i=1}\prod^n_{j=1}\dif z_{ij}$ and $\dif z
=\dif\re(z)\dif\im(z)$ for $z\in\complex$. With the matrix delta
function, we can rewrite the expression \eqref{eq:dist-for-W} as
\cite{Forrester2010}:
\begin{eqnarray}
\mathbf{Q}(W) = \frac1{2^m\pi^{m^2}}\int[\dif
T]\frac{\exp\Pa{\mathrm{i}\Tr{TW}}}{\det^n(\I+\mathrm{i}T)}.
\end{eqnarray}

\subsection{The distribution of diagonal part of Wishart ensemble}

Now break up $W=[w_{ij}]$ as $W=W^\diag+W^{\mathrm{off}}$, where
$W^\diag$ and $W^{\mathrm{off}}$ are the diagonal part and
off-diagonal part of $W$, respectively. Clearly this decomposition
is orthogonal with respect to Hilbert-Schmidt inner product in
operator space. The distribution of diagonal entries of Wishart
ensemble can be calculated via the following integral:
\begin{eqnarray*}
&&\mathbf{q}(W^\diag) = \int\mathbf{Q}(W)[\dif W^{\mathrm{off}}]\\
&&=\frac1{(2\pi)^m}\int[\dif T] \frac{\exp\Pa{\mathrm{i}\Tr{T^\diag
W^\diag}}}{\det^n(\I+\mathrm{i}T)} \Br{\frac1{\pi^{m(m-1)}}\int[\dif
W^{\mathrm{off}}]\exp\Pa{\mathrm{i}\Tr{T^{\mathrm{off}}W^{\mathrm{off}}}}}\\
&&=\frac1{(2\pi)^m}\int[\dif T]
\delta(T^{\mathrm{off}})\frac{\exp\Pa{\mathrm{i}\Tr{T^\diag
W^\diag}}}{\det^n(\I+\mathrm{i}T)}.
\end{eqnarray*}
Note that, in the third equality, we use the fact that
\begin{eqnarray}
\delta(T^{\mathrm{off}}) = \frac1{\pi^{m(m-1)}}\int[\dif
W^{\mathrm{off}}]\exp\Pa{\mathrm{i}\Tr{W^{\mathrm{off}}T^{\mathrm{off}}}}.
\end{eqnarray}
Now we have obtained that
\begin{eqnarray}\label{eq:diagonal-part-dist}
\mathbf{q}(W^\diag) = \frac1{(2\pi)^m}\int[\dif T]
\delta(T^{\mathrm{off}})\frac{\exp\Pa{\mathrm{i}\Tr{T^\diag
W^\diag}}}{\det^n(\I+\mathrm{i}T)}.
\end{eqnarray}
Next, we calculate the integral in \eqref{eq:diagonal-part-dist}. To
this end, denote
\begin{eqnarray}\label{eq:diagonal-integral}
\cI_{m,n}(W^\diag_m) := \int[\dif T_m]
\delta(T^{\mathrm{off}}_m)\frac{\exp\Pa{\mathrm{i}\Tr{T^\diag_m
W^\diag_m}}}{\det^n(\I_m+\mathrm{i}T_m)}
\end{eqnarray}
and we follow the approach taken by Janik \cite{Janik2003},
partition $T_m$ as $2\times2$ block matrix:
\begin{eqnarray*}
T_m = \Br{\begin{array}{cc}
            T_{m-1} & \ket{u} \\
            \bra{u} & t_{mm}
          \end{array}
}.
\end{eqnarray*}
Thus
\begin{eqnarray*}
\I_m+\mathrm{i}T_m = \Br{\begin{array}{cc}
            \I_{m-1}+\mathrm{i}T_{m-1} & \mathrm{i}\ket{u} \\
            \mathrm{i}\bra{u} & 1+\mathrm{i}t_{mm}
          \end{array}
}.
\end{eqnarray*}
By employing \emph{Schur determinant formula} \cite{Zhang2005}, we
get that
\begin{eqnarray}
\det(\I_m+\mathrm{i}T_m) =
\det(\I_{m-1}+\mathrm{i}T_{m-1})\Br{\mathrm{i}t_{mm} +1+
\Innerm{u}{\Pa{\I_{m-1}+\mathrm{i}T_{m-1}}^{-1}}{u}}.
\end{eqnarray}
Apparently,
\begin{eqnarray}
[\dif T_m]&=&[\dif T_{m-1}][\dif u]\dif
t_{mm},\\
\delta(T^{\mathrm{off}}_m) &=&
\delta(T^{\mathrm{off}}_{m-1})\delta(u),\\
\exp\Pa{\mathrm{i}\Tr{T^\diag_m W^\diag_m}} &=&
\exp\Pa{\mathrm{i}\Tr{T^\diag_{m-1}
W^\diag_{m-1}}}\exp\Pa{\mathrm{i}t_{mm} w_{mm}}.
\end{eqnarray}
Now \eqref{eq:diagonal-integral} can be transformed into the
following form:
\begin{eqnarray}
\cI_{m,n}(W^\diag_m) &=& \int[\dif T_{m-1}]
\delta(T^{\mathrm{off}}_{m-1})\frac{\exp\Pa{\mathrm{i}\Tr{T^\diag_{m-1}
W^\diag_{m-1}}}}{\det^n(\I_{m-1}+\mathrm{i}T_{m-1})}\int[\dif
u]\delta(u)\notag\\
&&\times\int\dif t_{mm} \frac{\exp\Pa{\mathrm{i}t_{mm}
w_{mm}}}{\Pa{{\mathrm{i}t_{mm} +1+
\Innerm{u}{\Pa{\I_{m-1}+\mathrm{i}T_{m-1}}^{-1}}{u}}}^n}.
\end{eqnarray}
Using the following complex integral formula:
\begin{eqnarray}
\int^{+\infty}_{-\infty}\dif x \frac{e^{\mathrm{i}xt}}{(x-a)^n} =
\frac{2\pi\mathrm{i}^n}{\Gamma(n)}t^{n-1}e^{\mathrm{i}ta}
\end{eqnarray}
or
\begin{eqnarray}
\int^{+\infty}_{-\infty}\dif x
\frac{e^{\mathrm{i}xt}}{(\mathrm{i}x+a)^\nu} =
\frac{2\pi}{\Gamma(\nu)}t^{\nu-1}e^{-ta},
\end{eqnarray}
we have
\begin{eqnarray}
&&\int\dif t_{mm} \frac{\exp\Pa{\mathrm{i}t_{mm}
w_{mm}}}{\Pa{{\mathrm{i}t_{mm} +1+
\Innerm{u}{\Pa{\I_{m-1}+\mathrm{i}T_{m-1}}^{-1}}{u}}}^n}\\
&&=\frac1{\mathrm{i}^n}\int^{+\infty}_{-\infty}\dif
t_{mm}\frac{\exp\Pa{\mathrm{i}t_{mm} w_{mm}}}{\Pa{{t_{mm}
-\mathrm{i}\Pa{1+
\Innerm{u}{\Pa{\I_{m-1}+\mathrm{i}T_{m-1}}^{-1}}{u}}}}^n}\\
&&=\frac1{\mathrm{i}^n}
\frac{2\pi\mathrm{i}^n}{\Gamma(n)}w_{mm}^{n-1}\exp\Br{-w_{mm}\Pa{1+
\Innerm{u}{\Pa{\I_{m-1}+\mathrm{i}T_{m-1}}^{-1}}{u}}},
\end{eqnarray}
that is,
\begin{eqnarray}
&&\int\dif t_{mm} \frac{\exp\Pa{\mathrm{i}t_{mm}
w_{mm}}}{\Pa{{\mathrm{i}t_{mm} +1+
\Innerm{u}{\Pa{\I_{m-1}+\mathrm{i}T_{m-1}}^{-1}}{u}}}^n}\\
&&= \frac{2\pi}{\Gamma(n)}w_{mm}^{n-1}e^{-w_{mm}}\exp\Pa{-w_{mm}
\Innerm{u}{\Pa{\I_{m-1}+\mathrm{i}T_{m-1}}^{-1}}{u}}.
\end{eqnarray}
Then
\begin{eqnarray}
\cI_{m,n}(W^\diag_m) &=&
\frac{2\pi}{\Gamma(n)}w_{mm}^{n-1}e^{-w_{mm}}\int[\dif T_{m-1}]
\delta(T^{\mathrm{off}}_{m-1})\frac{\exp\Pa{\mathrm{i}\Tr{T^\diag_{m-1}
W^\diag_{m-1}}}}{\det^n(\I_{m-1}+\mathrm{i}T_{m-1})}\notag\\
&&\times\int[\dif u]\delta(u)\exp\Pa{-w_{mm}
\Innerm{u}{\Pa{\I_{m-1}+\mathrm{i}T_{m-1}}^{-1}}{u}}.
\end{eqnarray}
Therefore
\begin{eqnarray}
\cI_{m,n}(W^\diag_m) =
\frac{2\pi}{\Gamma(n)}w_{mm}^{n-1}e^{-w_{mm}}\int[\dif T_{m-1}]
\delta(T^{\mathrm{off}}_{m-1})\frac{\exp\Pa{\mathrm{i}\Tr{T^\diag_{m-1}
W^\diag_{m-1}}}}{\det^n(\I_{m-1}+\mathrm{i}T_{m-1})},
\end{eqnarray}
where we used the fact that
\begin{eqnarray}
\int[\dif u]\delta(u)\exp\Pa{-w_{mm}
\Innerm{u}{\Pa{\I_{m-1}+\mathrm{i}T_{m-1}}^{-1}}{u}} = 1.
\end{eqnarray}
That is,
\begin{eqnarray}
\cI_{m,n}(W^\diag_m) =
\frac{2\pi}{\Gamma(n)}w_{mm}^{n-1}e^{-w_{mm}}\cI_{m-1,n}(W^\diag_{m-1}).
\end{eqnarray}
By this, we finally get
\begin{eqnarray}
\cI_{m,n}(W^\diag_m) =
\frac{(2\pi)^m}{\Gamma(n)^m}\Pa{\prod^m_{j=1}w_{jj}^{n-1}}\exp\Pa{-\sum^m_{j=1}w_{jj}}.
\end{eqnarray}
In summary, we have the distribution of diagonal part of Wishart
ensemble, which can be presented as the following proposition.
\begin{prop}[The distribution of diagonal part of Wishart
ensemble] The distribution of diagonal part of Wishart ensemble is
given by the following formula:
\begin{eqnarray}
\mathbf{q}(W^\diag) =
\frac1{\Gamma(n)^m}\Pa{\prod^m_{j=1}w_{jj}^{n-1}}\exp\Pa{-\sum^m_{j=1}w_{jj}}.
\end{eqnarray}
\end{prop}

\begin{remark}
In the above process, we calculate from first principle the joint
distribution of diagonal part of Wishart matrices. What we emphasize
here is that the distribution of diagonal part of Wishart matrices
is the marginal distribution of the distribution of Wishart
matrices. Of course, we can derive directly this result from the
definition of Wishart matrices.
\end{remark}

\subsection{Wishart ensemble and random mixed quantum states}

For the mathematical treatment of a quantum system, one usually
associates with it a Hilbert space whose vectors describe the states
of that system. In our situation, we associate with $A$ and $B$ two
complex Hilbert spaces $\cH_A$ and $\cH_B$ with respective
dimensions $m$ and $n$, which are assumed here to be such that $m
\leqslant n$. In these settings, the vectors of the spaces $\cH_A$
and $\cH_B$ describe the states of the systems $A$ and $B$. Those of
the tensorial product $\cH_A\ot\cH_B$ (of dimension $mn$) then
describe the states of the combined system $AB$.

It will be helpful throughout this paper to make use of a simple
correspondence between the linear operator spaces $\lin{\cX,\cY}$
and $\cY\ot\cX$, for given complex Euclidean spaces $\cX$ and $\cY$.
We define the mapping
\begin{eqnarray*}
\vec:\lin{\cX,\cY}\to \cY\ot\cX
\end{eqnarray*}
to be the linear mapping that represents a change of bases from the
standard basis of $\lin{\cX,\cY}$ to the standard basis of
$\cY\ot\cX$. Specifically, in the Dirac notation, this mapping
amounts to flipping a bra to a ket, we define
\begin{eqnarray*}
\vec(\out{i}{j}) = \ket{ij},
\end{eqnarray*}
at this point the mapping is determined for every
$M=\sum_{i,j}M_{ij}\out{i}{j}\in\lin{\cX,\cY}$ by linearity
\cite{watrous2016preprint}. For convenience, we denote
$\vec(M):=\ket{M}$. Clearly $\ptr{\cX}{\out{M}{N}} = MN^\dagger$ for
$M,N\in\lin{\cX,\cY}$. In this problem, we assume that
$\cX=\cH_B=\complex^n$ and $\cY=\cH_A=\complex^m$.

For Wishart matrices $W$ on $\complex^m$, it can be considered as
the reduced state of a purified random vector $\ket{Z}$ with random
coordinates $z_{ij},i=1,\ldots,m;j=1,\ldots,n$, the probability
distribution of which being the uniform distribution on the unit
sphere of $\complex^m\ot\complex^n$. That is, $W=ZZ^\dagger =
\Ptr{\complex^n}{\out{Z}{Z}}$. Since any random mixed quantum states
can be generated by Wishart matrices with a fixed trace one, it
follows that the distribution of random mixed quantum states is
given by the following via $\rho=\frac{W}{\Tr{W}}$:
\begin{eqnarray}
\mathbf{P}(\rho) \propto \delta(1-\Tr{W})\mathbf{Q}(W).
\end{eqnarray}
The distribution of random mixed quantum states is given by the
following:
\begin{eqnarray}\label{eq:wishar-of-fixed-trace}
\mathbf{P}(\rho) \propto \delta(1-\Tr{\rho})\int[\dif
Z]\delta(\rho-ZZ^\dagger)\varphi(Z).
\end{eqnarray}
That is \cite{Karol2001},
\begin{eqnarray}
\mathbf{P}(\rho) \propto \delta(1-\Tr{\rho})\det^{n-m}(\rho).
\end{eqnarray}
In view of this, we get
\begin{prop}
The distribution of diagonal part $\rho_\diag$ of random mixed
quantum states is given by the following:
\begin{eqnarray}
\mathbf{p}(\rho_\diag) =
\int[\dif\rho_{\mathrm{off}}]\mathbf{P}(\rho)=
\frac{\Gamma(mn)}{\Gamma(n)^m}\delta\Pa{1-\sum^m_{j=1}\rho_{jj}}\prod^m_{j=1}\rho_{jj}^{n-1}.
\end{eqnarray}
\end{prop}
We check it directly. A random reduced quantum state $\rho$,
obtained by partial tracing a Haar-distributed bipartite state
$\ket{Z}\in \complex^m\ot\complex^n(m\leqslant n)$, can be expressed
via a Wishart matrix as follows:
$$
\rho=\frac{W}{\Tr{W}}=\frac1tW,
$$
where $W=ZZ^\dagger$ for $Z=[z_{ij}]$ is an $m\times n$ matrix with
independent Gaussian complex entries, and
$t:=\Tr{W}=\Tr{ZZ^\dagger}$. Then $t\rho=W=ZZ^\dagger$, i.e.
$$
W_{ii}=t\rho_{ii} = \sum^n_{j=1}\abs{z_{ij}}^2 =
\sum^n_{j=1}\Br{\Pa{\re(z_{ij})}^2+\Pa{\im(z_{ij})}^2},
$$
where $\re(z_{ij}),\im(z_{ij})\sim N(0,\frac12)$, leading to the
following: for all $i=1,\ldots,m$,
$$
p\Pa{W_{ii}} = \frac{W_{ii}^{n-1}}{2^n\Gamma(n)}e^{-\frac12W_{ii}},
$$
leading to the following:
$$
p(t) = \frac{t^{mn-1}}{2^{mn}\Gamma(mn)}e^{-\frac12t}.
$$
Let us perform the following change of variables:
$\Pa{W_{11},\ldots,W_{mm}}\mapsto (
\rho_{11},\ldots,\rho_{m-1,m-1},t)$. The Jacobian of the
transformation \cite{mathai1997} is
$$
\abs{\frac{\partial \Pa{W_{11},\ldots,W_{mm}}}{\partial
(\rho_{11},\ldots,\rho_{m-1,m-1},t)}} = t^{m-1}.
$$
That is
$$
\dif W_{11}\cdots \dif W_{mm}=t^{m-1}\dif t\dif \rho_{11}\cdots \dif
\rho_{m-1,m-1}.
$$
Furthermore,
\begin{eqnarray*}
&&p\Pa{W_{11}}\cdots p\Pa{W_{mm}}\dif W_{11}\cdots\dif
W_{mm}\\
&&=p\Pa{t\rho_{11}}\cdots p\Pa{t\rho_{mm}}t^{m-1}\dif t\dif
\rho_{11}\cdots \dif \rho_{m-1,m-1}\\
&&=\frac{\Pa{t\rho_{11}}^{n-1}}{2^n\Gamma(n)}e^{-\frac12t\rho_{11}}
\times\cdots\times
\frac{\Pa{t\rho_{mm}}^{n-1}}{2^n\Gamma(n)}e^{-\frac12t\rho_{mm}}t^{m-1}\dif
t\dif \rho_{11}\cdots
\dif \rho_{m-1,m-1}\\
&&=\frac{t^{mn-m}}{2^{mn}\Gamma(n)^m}e^{-\frac12t}\prod^m_{j=1}\rho^{n-1}_{jj}\times
t^{m-1}\dif t\dif \rho_{11}\cdots \dif \rho_{m-1,m-1}.
\end{eqnarray*}
Finally we get
\begin{eqnarray*}
&&p\Pa{W_{11}}\cdots p\Pa{W_{mm}}\dif W_{11}\cdots\dif
W_{mm}=p\Pa{W_{11},\ldots,W_{mm}}\dif W_{11}\cdots\dif
W_{mm}\\
&&=p\Pa{t\rho_{11},\ldots,t\rho_{mm}}t^{m-1}\dif t\dif
\rho_{11}\cdots \dif \rho_{m-1,m-1}\\
&&=\frac{t^{mn-1}}{2^{mn}\Gamma(mn)}e^{-\frac12t}\dif
t\times\frac{\Gamma(mn)}{\Gamma(n)^m}\prod^m_{j=1}\rho^{n-1}_{jj}\dif
\rho_{11}\cdots \dif \rho_{m-1,m-1}.
\end{eqnarray*}
Taking the integration with respect to $t$ gives rise to the
marginal distribution---the distribution of the diagonal elements
which is the symmetric Dirichlet distribution \cite{mathai1997}:
\begin{eqnarray*}
\mathbf{p}(\rho_\diag) := \mathbf{p}(\rho_{11},\ldots,\rho_{mm}) =
\frac{\Gamma(mn)}{\Gamma(n)^m}\delta\Pa{1-\sum^m_{j=1}\rho_{jj}}\prod^m_{j=1}\rho^{n-1}_{jj}.
\end{eqnarray*}
The following result, although beyond our goal in the present paper,
we still record it here for independent interests, dealt with the
exact analytical relationship between the joint distributions of
diagonal entries and eigenvalues of the same invariant ensemble.

\begin{prop}[Derivative principle, \cite{zapata}]\label{prop:der-princ}
Let $Z$ be a random matrix drawn from a unitarily invariant random
matrix ensemble, $\varrho_Z$ the joint eigenvalue distribution for
$Z$ and $\mathbf{p}_Z$ the joint distribution of the diagonal
elements of $Z$. Then
\begin{eqnarray}\label{eq:derivative-principle}
\varrho_Z(\lambda) &=&
\frac1{\prod^N_{k=1}k!}\Delta(\lambda)\Delta(-\partial_\lambda)\mathbf{p}_Z(\lambda),
\end{eqnarray}
where $\Delta(\lambda) = \prod_{i<j}(\lambda_j-\lambda_i)$ is the
Vandermonde determinant and $\Delta(-\partial_\lambda)$ the
differential operator
$\prod_{i<j}\Pa{\frac{\partial}{\partial_{\lambda_i}} -
\frac{\partial}{\partial_{\lambda_j}}}$.
\end{prop}

\subsection{Average entropy of the diagonal entries of random density matrices}

In what follows, we calculate the average entropy of the diagonal
part of random density matrices under the distribution of random
density matrices subject to \eqref{eq:wishar-of-fixed-trace}.
Specifically, we will calculate the following integral:
\begin{eqnarray}
\overline{S}^D_{m,n} &=& \int
S(\rho_\diag)\mathbf{P}(\rho)[\dif\rho] = \int[\dif\rho_\diag]
S(\rho_\diag)\int
[\dif\rho_{\mathrm{off}}]\mathbf{P}(\rho)\notag\\
&=& \int[\dif\rho_\diag]
S(\rho_\diag)\mathbf{p}(\rho_\diag),\label{eq:Dirichlet-int}
\end{eqnarray}
where
\begin{eqnarray*}
\mathbf{P}(\rho) \propto \delta(1-\Tr{\rho})\int[\dif Z]\delta(\rho
- ZZ^\dagger)\varphi(Z)
\end{eqnarray*}
for $\varphi(Z) = \frac1{\pi^{mn}}\exp\Pa{-\Tr{ZZ^\dagger}}$. We
have the following result:
\begin{prop}
The average diagonal entropy of random mixed quantum states, induced
by Haar-distributed bipartite pure states on $\complex^m\ot
\complex^n$, is given by the following:
\begin{eqnarray}
\overline{S}^D_{m,n} = H_{mn} - H_n,
\end{eqnarray}
where $H_k := \sum^k_{j=1}\frac1j$ is the $k$-th harmonic number for
positive integer number $k$.
\end{prop}

\begin{proof}
According the distribution of diagonal part of random mixed quantum
states, we have
\begin{eqnarray*}
\overline{S}^D_{m,n} &=& \int S(\rho_\diag)\mathbf{p}(\rho_\diag)[\dif\rho_\diag]\\
&=&\frac{\Gamma(mn)}{\Gamma(n)^m}\int \Pa{-\sum^m_{j=1}\rho_{jj}\ln
\rho_{jj}}
\delta\Pa{1-\sum^m_{j=1}\rho_{jj}}\prod^m_{j=1}\rho^{n-1}_{jj}\dif \rho_{jj}\\
&=& -m\frac{\Gamma(mn)}{\Gamma(n)^m}\int (\rho_{11}\ln
\rho_{11})\delta\Pa{1-\sum^m_{j=1}\rho_{jj}}\prod^m_{j=1}\rho^{n-1}_{jj}\dif
\rho_{jj}.
\end{eqnarray*}
Then it can be rewritten as
\begin{eqnarray}
\overline{S}^D_{m,n} &=&
-m\frac{\Gamma(mn)}{\Gamma(n)^m}\int^1_0\dif \rho_{11}
\rho^n_{11}\ln \rho_{11}\int^\infty_0\cdots\int^\infty_0
\delta\Pa{(1-\rho_{11})-\sum^m_{j=2}\rho_{jj}}\prod^m_{j=2}\rho^{n-1}_{jj}\dif
\rho_{jj}.
\end{eqnarray}
Now denote
$$
F(t) = \int^\infty_0\cdots\int^\infty_0
\delta\Pa{t-\sum^m_{j=2}\rho_{jj}}\prod^m_{j=2}\rho^{n-1}_{jj}\dif
\rho_{jj}.
$$
Then via $y_j=s\rho_{j+1,j+1}$, where $j=1,\ldots,m-1$, performing
Laplace transform $(t\to s)$ to $F(t)$ \cite{Williams1973}, we get
\begin{eqnarray*}
\widetilde F(s) &=& \int^\infty_0\cdots\int^\infty_0
\exp\Pa{-s\sum^m_{j=2}\rho_{jj}}\prod^m_{j=2}\rho^{n-1}_{jj}\dif
\rho_{jj} \\
&=& s^{-n(m-1)}
\int^\infty_0\cdots\int^\infty_0\prod^{m-1}_{j=1}e^{-y_j}y^{n-1}_j\dif
y_j\\
&=&s^{-n(m-1)} \prod^{m-1}_{j=1}\int^\infty_0y^{n-1}_je^{-y_j}\dif
y_j =s^{-n(m-1)} \prod^{m-1}_{j=1}\Gamma(n),
\end{eqnarray*}
that is, $\widetilde F(s) = \Gamma(n)^{m-1}s^{-n(m-1)}$. This
implies
$$
F(t) = \frac{\Gamma(n)^{m-1}}{\Gamma(n(m-1))}t^{n(m-1)-1}.
$$
Therefore
\begin{eqnarray*}
\overline{S}^D_{m,n} &=& -m\frac{\Gamma(mn)}{\Gamma(n)^m}\int^1_0\dif \rho_{11} \rho^n_{11}\ln \rho_{11}F(1-\rho_{11})\\
&=&-m\frac{\Gamma(mn)}{\Gamma(n)^m}\frac{\Gamma(n)^{m-1}}{\Gamma(n(m-1))}\int^1_0
x^n\ln x
(1-x)^{n(m-1)-1}\dif x\\
&=&-m\frac{\Gamma(mn)}{\Gamma(n)^m}\frac{\Gamma(n)^{m-1}}{\Gamma(n(m-1))}\left.\frac{\partial
B}{\partial \alpha}\right|_{(\alpha, \beta)=(n+1,n(m-1))},
\end{eqnarray*}
where $B(\alpha,\beta)=\int^1_0x^{\alpha-1}(1-x)^{\beta-1}\dif
x=\frac{\Gamma(\alpha)\Gamma(\beta)}{\Gamma(\alpha+\beta)}$ and
$$
\frac{\partial B}{\partial \alpha}
=\frac{\Gamma(\alpha)\Gamma(\beta)}{\Gamma(\alpha+\beta)}(\psi(\alpha)-\psi(\alpha+\beta)),
$$
where $\psi(\alpha):=\frac{\dif}{\dif \alpha}\ln\Gamma(\alpha)$.
Letting $(\alpha,\beta)=(n+1,n(m-1))$ give rises to
\begin{eqnarray}
\overline{S}^D_{m,n} =
-m\frac{\Gamma(mn)}{\Gamma(n)^m}\frac{\Gamma(n)^{m-1}}{\Gamma(n(m-1))}
\frac{\Gamma(n+1)\Gamma(n(m-1))}{\Gamma(mn+1)}(\psi(n+1)-\psi(mn+1)),
\end{eqnarray}
where $\psi(n+1)=H_n-\gamma_{\text{Euler}}$ for Euler's constant
$\gamma_{\text{Euler}}\approx0.57722$ \cite{Young1991}. That is,
\begin{eqnarray}
\overline{S}^D_{m,n} = \psi(mn+1)-\psi(n+1) = H_{mn} - H_n.
\end{eqnarray}
We are done.
\end{proof}
Note that similar integrals like the one in \eqref{eq:Dirichlet-int}
are considered recently for the motivation from machine-learning,
see \cite{Montufar}. This result is very interesting, compared with
Page's formula \cite{FK1994}, stating the average entropy of a
subsystem given by $\overline{S}_{m,n}=H_{mn}-H_n -
\frac{m-1}{2n}(m\leqslant n)$. With this result, we can give the
average relative entropy of coherence for random mixed quantum
states, obtained recently in the paper \cite{Zhang2015}, in the
following section.

\section{Main results}\label{sect:main-results}

In this section, we will present our main result about quantum
coherence, stating the average relative entropy of coherence for
random mixed quantum states can be given by the following compact
formula (see also in \cite{Zhang2015}). Note that what we emphasize
here is the method used for deriving this elegant formula.

\begin{thrm}[Average coherence]\label{th:ave-coh}
For random mixed states of dimension $m$ sampled from induced
measures obtained via partial tracing of Haar distributed bipartite
pure states of dimension $mn$ where $m\leqslant n$, the average
relative entropy of coherence is given by the following compact form
\begin{eqnarray}\label{eq:h-1th}
\overline{\sC}_{m,n} = \frac{m-1}{2n}.
\end{eqnarray}
\end{thrm}

\begin{proof}
Now the distribution of random mixed quantum states is given by
\begin{eqnarray}
\mathbf{P}(\rho)\propto \delta(1-\Tr{\rho})\det^{n-m}(\rho).
\end{eqnarray}
Under this distribution, we calculate the average relative entropy
of coherence as follows:
\begin{eqnarray}
&&\overline{\sC}_{m,n} = \int[\dif
\rho]\mathbf{P}(\rho)\sC_{m,n}(\rho) = \int[\dif
\rho]\mathbf{P}(\rho)S(\rho_\diag) - \int[\dif
\rho]\mathbf{P}(\rho)S(\rho)\\
&& = \overline{S}^D_{m,n} - \overline{S}_{m,n} = (H_{mn} - H_n) -
\Pa{H_{mn} - H_n - \frac{m-1}{2n}} = \frac{m-1}{2n}.
\end{eqnarray}
Note here that we used the fact that
\begin{eqnarray}
\overline{S}_{m,n} = \int[\dif \rho]\mathbf{P}(\rho)S(\rho) = H_{mn}
- H_n - \frac{m-1}{2n}
\end{eqnarray}
which is called Page's average entropy formula, conjectured in
\cite{Page1993}, and proven in \cite{FK1994,SR1995,Sen1996}.
\end{proof}

\begin{remark}
For $m=n$, we see that the average coherence is given by
$\frac{m-1}{2m}$, which is approaching to $\frac12$ when
$m\to\infty$. The asymptotic value $\frac12$ of the average
coherence is confirmed by Pucha{\l}a \emph{et. al} using tools from
free probability theory \cite{Puchala2015}.
\end{remark}

We have already known that the distribution of random mixed quantum
states is given by
\begin{eqnarray*}
\mathbf{P}(\rho)\propto \delta(1-\Tr{\rho})\det^{n-m}(\rho).
\end{eqnarray*}
By the spectral decomposition of $\rho$, we have $\rho=U\Lambda
U^\dagger$ with $\Lambda=\diag(\lambda_1,\ldots,\lambda_m)$, where
$\lambda_i\neq\lambda_j$ for distinct indices $i$ and $j$. Since
this distribution $\mathbf{P}(\rho)$ is unitary-invariant, noting
that
$$
[\dif\rho]\propto \prod_{1\leqslant i<j\leqslant
m}(\lambda_j-\lambda_i)^2[\dif\Lambda]\dif\mu_{\mathrm{Haar}}(U),
$$
it follows that the joint distribution of the eigenvalues of random
mixed quantum states is given by \cite{Karol2001}
\begin{eqnarray*}
P_{m,n}(\Lambda)\propto
\delta\Pa{1-\sum^m_{j=1}\lambda_j}\prod_{1\leqslant i<j\leqslant
m}(\lambda_j-\lambda_i)^2\prod^m_{j=1}\lambda^{n-m}_j.
\end{eqnarray*}
In what follows, we reconsider the calculation of the average
coherence of random mixed quantum states:
\begin{eqnarray*}
\overline{\sC}_{m,n} = \int[\dif\rho]\mathbf{P}(\rho)\sC_{m,n}(\rho)
=
\int[\dif\Lambda]P_{m,n}(\Lambda)\Br{\int\dif\mu_{\mathrm{Haar}}(U)\sC_{m,n}(U\Lambda
U^\dagger)}.
\end{eqnarray*}
Clearly $\sC_{m,n}(U\Lambda U^\dagger) = S((U\Lambda
U^\dagger)_\diag) - S(U\Lambda U^\dagger) = S((U\Lambda
U^\dagger)_\diag) - S(\Lambda)$, which implies that
\begin{eqnarray*}
\overline{\sC}_{m,n} =
\int[\dif\Lambda]P_{m,n}(\Lambda)\overline{\sC}^{\mathrm{iso}}_{m,n}(\Lambda),
\end{eqnarray*}
where
\begin{eqnarray*}
\overline{\sC}^{\mathrm{iso}}_{m,n}(\Lambda) = \int S((U\Lambda
U^\dagger)_\diag)\dif\mu_{\mathrm{Haar}}(U) - S(\Lambda).
\end{eqnarray*}
Earlier in the study of upper bounds and lower bounds of classical
accessible information, one obtains that the average diagonal
entropy of isospectral quantum state with a fixed spectrum
$\Lambda=\set{\lambda_1,\ldots,\lambda_m}$ is given by
\cite{Jones1991,Jozsa1994}:
\begin{eqnarray}
\int S((U\Lambda U^\dagger)_\diag)\dif\mu_{\mathrm{Haar}}(U) = H_m -
1 + Q(\Lambda),
\end{eqnarray}
where
\begin{eqnarray}
Q(\Lambda) := -\sum^m_{i=1}
\frac{\lambda^m_i\ln\lambda_i}{\prod_{j\in\widehat
i}(\lambda_i-\lambda_j)}\quad (\widehat
i:=\set{1,\ldots,m}\backslash\set{i}),
\end{eqnarray}
which is called \emph{subentropy} \cite{Datta2014,Jozsa1994}. From
this, we see that the average relative entropy of coherence of
isospectral quantum states of fixed spectrum
$\Lambda=\set{\lambda_1,\ldots,\lambda_m}$ is obtained easily
\begin{eqnarray}
\overline{\sC}^{\mathrm{iso}}_{m,n}(\Lambda) = \int S((U\Lambda
U^\dagger)_\diag)\dif\mu_{\mathrm{Haar}}(U) - S(\Lambda) = H_m - 1 +
Q(\Lambda) - S(\Lambda).
\end{eqnarray}
Now
\begin{eqnarray}
\overline{\sC}_{m,n} = \int[\dif\Lambda]P_{m,n}(\Lambda)
\overline{\sC}^{\mathrm{iso}}_{m,n}(\Lambda) = H_m - 1 +
\int[\dif\Lambda] P_{m,n}(\Lambda) Q(\Lambda) -
\int[\dif\Lambda]P_{m,n}(\Lambda) S(\Lambda).
\end{eqnarray}
By using Page's formula for the average von Neumann entropy:
\begin{eqnarray}
\int[\dif\Lambda]P_{m,n}(\Lambda) S(\Lambda) =\overline{S}_{m,n}=
H_{mn}- H_n -\frac{m-1}{2n}.
\end{eqnarray}
If we denote
\begin{eqnarray}
\overline{Q}_{m,n}:=\int[\dif\rho] \mathbf{P}(\rho) Q(\rho),
\end{eqnarray}
then
\begin{eqnarray}
\overline{Q}_{m,n}&=&\int[\dif\Lambda] P_{m,n}(\Lambda) \int
\dif\mu_{\mathrm{Haar}}(U)Q(U\Lambda U^\dagger)\notag\\
&=& \int[\dif\Lambda] P_{m,n}(\Lambda) Q(\Lambda),
\end{eqnarray}
then using the result of the average relative entropy of coherence
for random mixed quantum states, we get the following result which
may be of independent interest later in the investigation of quantum
information theory.
\begin{thrm}[Average subentropy]\label{th:ave-suben}
For random mixed states of dimension $m$ sampled from induced
measures obtained via partial tracing of Haar distributed bipartite
pure states of dimension $mn$ where $m\leqslant n$, the average
relative entropy of coherence is given by
\begin{eqnarray}
\overline{Q}_{m,n} = 1+H_{mn}-H_m-H_n.
\end{eqnarray}
\end{thrm}

\begin{remark}
Note that in \cite{Zhang2015} (please find the meaning of
corresponding notations therein), the authors have obtained that
\begin{eqnarray}
\cI^Q_m(n-m+1,1) =
\frac1{mn}\sum^{m-1}_{k=0}\frac{(-1)^k\Gamma(m+n-k)}{k!\Gamma(m-k)\Gamma(n-k)}\Br{\psi(mn+1)-\psi(m+n-k)}.
\end{eqnarray}
In our notation here, $\cI^Q_m(n-m+1,1)
=\overline{Q}_{m,n}=1+H_{mn}-H_m-H_n$. The above compact form of the
average subentropy of random mixed quantum states can be rewritten
as
\begin{eqnarray}
\overline{Q}_{m,n} = (1-\gamma_{\mathrm{Euler}}) - (a_m+a_n-a_{mn}),
\end{eqnarray}
where $a_k:=H_k-\ln k -\gamma_{\mathrm{Euler}}$. We see that
$\lim_{k\to\infty}a_k=0$. Moreover, since $m\leqslant n$, it follows
that $\lim_{m\to\infty}(a_m+a_n-a_{mn})=0$. Therefore we obtain that
\begin{eqnarray}
\lim_{m\to\infty}\overline{Q}_{m,n} = 1-\gamma_{\mathrm{Euler}}.
\end{eqnarray}
Here we have given an analytical proof about the fact that the
subentropy is a nonlinear function of random mixed quantum states,
in higher dimensional space, the average subentropy approaches the
maximal value of the subentropy where it is taken at only maximally
mixed state. This amounts to say, in higher dimensional space the
following identity holds approximately
\begin{eqnarray}
Q\Pa{\int \rho\mathbf{P}(\rho)[\dif\rho]} \simeq \int
Q(\rho)\mathbf{P}(\rho)[\dif\rho].
\end{eqnarray}
\end{remark}

\begin{remark}
We make remark here about approximations of $\overline{S}_{m,n}$ and
$\overline{Q}_{m,n}$. We know that $S_{\max}=\ln m$ and
$Q_{\max}=(1-\gamma_{\mathrm{Euler}})- a_m\doteq
1-\gamma_{\mathrm{Euler}}$ \cite{Jozsa1994}. Then
\begin{eqnarray}
S_{\max} - \overline{S}_{m,n} &=& \frac{m-1}{2n}+(a_n-a_{mn}),\\
Q_{\max} - \overline{Q}_{m,n} &=& a_n-a_{mn}.
\end{eqnarray}
Thus for the ratio $\frac mn$ fixed, $S_{\max} - \overline{S}_{m,n}$
approaches a nonzero constant, whereas $Q_{\max} -
\overline{Q}_{m,n}$ approaches zero. However, this might not be too
surprising, since $S_{\max} = \ln m$ grows indefinitely with $m$,
whereas $Q_{\max}$ is a constant. Thus the relative errors of
$S_{\max}$ and of $Q_{\max}$ as approximations for the mean values
$\overline{S}_{m,n}$ and $\overline{Q}_{m,n}$ both tend to zero for
large $m$:
\begin{eqnarray}
\frac{S_{\max} - \overline{S}_{m,n}}{S_{\max}} &=& \frac{m-1}{2n\ln m}+\frac{a_n-a_{mn}}{\ln m}\sim \frac{m}{\ln m}\cdot\frac1{2n},\\
\frac{Q_{\max} - \overline{Q}_{m,n}}{Q_{\max}}
&=&\frac{a_n-a_{mn}}{1-\gamma_{\mathrm{Euler}}-a_m}\sim
\frac{1}{1-\gamma_{\mathrm{Euler}}}\cdot\frac1{2n}.
\end{eqnarray}
For fixed $m$, both of these go to zero inversely with $n$ when $n$
is taken to infinity.  However, for fixed ratio $\frac mn$, the
relative error of $S_{\max}$ goes to zero slower (as the inverse of
the logarithm of $m$ or $n$) than the relative error of $Q_{\max}$
(as the inverse of $m$ or $n$), so in that sense for large $m$ and
fixed $\frac mn$, $Q_{\max}$ is a relatively better approximation
for $\overline{Q}_{m,n}$ than $S_{\max}$ is for
$\overline{S}_{m,n}$.
\end{remark}
The typicality of coherence is already established in
\cite{Zhang2015} without the closed-form of the average coherence.
For completeness, we include it here.

\begin{thrm}[Typicality of coherence]\label{th:typ-coh}
Let $\rho$ be a random mixed state on an $m$-dimensional Hilbert
space, where $m\geqslant 3$, induced via partial-tracing of the
Haar-distributed bipartite pure states on $mn$-dimensional Hilbert
space. Then for all positive scalars $\varepsilon>0$, the coherence
$\sC_{m,n}(\rho)$ of $\rho$ satisfies the following inequality:
\begin{eqnarray}
\mathbb{P}\Set{\abs{\sC_{m,n}(\rho) -
\frac{m-1}{2n}}>\varepsilon}\leqslant
2\exp\Pa{-\frac{mn\varepsilon^2}{144\pi^3\ln2(\ln m)^2}}.
\end{eqnarray}
\end{thrm}
From the above result, we can see that the entropy difference
$S(\rho_\diag)-S(\rho)$ is centered around the fraction
$\frac{m-1}{2n}$ except a set of exponential small probability
whenever the dimension of system under consideration is large
enough. This explains quantitatively why the diagonal part of random
mixed quantum states being more disorder than the eigenvalue.

\section{Extension to mixing of random mixed quantum
states}\label{sect:mixing}

In this section, we consider the following problem: We choose
arbitrary two Haar-distributed bipartite states $\ket{\psi_1}$ and
$\ket{\psi_2}$ from $\complex^m\ot\complex^n(m\leqslant n)$. Choose
uniformly a weight $w\in[0,1]$. There exist two $m\times n$ complex
matrices $Z_1,Z_2$ such that
$\sqrt{w}\ket{\psi_1}=\ket{Z_1},\sqrt{1-w}\ket{\psi_2}=\ket{Z_2}$.
we can form a new state
$w\out{\psi_1}{\psi_1}+(1-w)\out{\psi_2}{\psi_2}$ from
$\out{\psi_1}{\psi_1}$ and $\out{\psi_2}{\psi_2}$. By partial
tracing over the second system space $\complex^n$, we get a random
mixed quantum state $\sigma=Z_1Z^\dagger_1+Z_2Z^\dagger_2$. In fact,
random quantum state ensemble $\rho_{\mathrm{mix}}$ can be realized
as Wishart matrix ensemble. Let $Z=[Z_1\ Z_2]$, which is clearly a
matrix of size $m\times (2n)$. Thus $\sigma=ZZ^\dagger$. From the
previous discussion, we see that the distribution of matrix elements
of $\sigma$ is given by
\begin{eqnarray}
\mathbf{P}(\sigma)\propto \delta(1-\Tr{\sigma})\det^{2n-m}(\sigma)
\end{eqnarray}
and the distribution of the diagonal part of $\sigma$ is given by
\begin{eqnarray}
\mathbf{p}\Pa{\sigma_\diag}\propto
\delta\Pa{1-\sum^m_{j=1}\sigma_{jj}}\prod^m_{j=1}\sigma^{2n-1}_{jj}.
\end{eqnarray}
Denote by $\cE_1$ the random mixed quantum state ensemble obtained
by partial tracing over the second system space $\complex^n$ of
Haar-distributed bipartite pure states; $\cE_2$ arbitrary
probabilistic mixing of \emph{two} random chosen quantum
states;...;$\cE_k$ arbitrary probabilistic mixing of $k$ random
chosen quantum states, where $k$ is an arbitrary positive integer.
Let $\overline{\sC}_{m,n}(\cE_k)$ be the average coherence of
ensemble $\cE_k$ over $\complex^m\ot\complex^n$. With these
notations, we see that random state $\sigma$ is from ensemble
$\cE_2$, thus
\begin{eqnarray}
\overline{\sC}_{m,n}(\cE_2) = \int[\dif
\sigma]\mathbf{P}(\sigma)(S(\sigma_\diag) - S(\sigma)) =
\frac{m-1}{4n}.
\end{eqnarray}
Therefore we can summarize our main results in the present paper as:
\begin{eqnarray}\label{eq:h-kth}
\overline{\sC}_{m,n}(\cE_k) =
\frac1k\cdot\frac{m-1}{2n}=\frac1k\overline{\sC}_{m,n}(\cE_1).
\end{eqnarray}
Note that \eqref{eq:h-1th} is just here
$\overline{\sC}_{m,n}(\cE_1)$, a special case where $k=1$ of
\eqref{eq:h-kth}. We see that mixing of random states changes the
way of distribution of diagonal parts and eigenvalues as well,
respectively. An interpretation of \eqref{eq:h-kth} maybe is: For
fixed $m,n$, when mixing times $k$ is larger, the average coherence
is less. This suggest one also that if one would like to enhance
quantum coherence of quantum states, then one need distill coherent
part by dropping the incoherent part of quantum states.

We also see that \eqref{eq:h-kth} confirms in the probabilistic
sense that the convexity requirement for coherence monotone is
reasonable. That is, mixing of quantum states decreases coherence.

\section{Concluding remarks}\label{sect:conclusion}

In this paper we have analyzed the properties of the reduced density
matrices obtained from a suitable ensemble of pure states, we spend
very few pages to extends our results concerning the statistical
behavior of quantum coherence and subentropy that are obtained by
much effort in \cite{Zhang2015}. The main contributions of this
paper are that we give a new approach to get the compact formulae
for the average coherence and the average subentropy. In the future
research, we will study the distribution function of quantum
coherence in order to get more elaborate results on coherence via
the method used in \cite{Nadal2010,Nadal2011}. We hope that the
methods and results in this paper can provide new light over the
related problems in quantum information theory.

\subsubsection*{Acknowledgements}
L.Z. would like to thank Don N. Page for his comments on the
approximations for the average entropy and subentropy. This work is
supported by Natural Science Foundation of Zhejiang Province of
China (LY17A010027) and by National Natural Science Foundation of
China (Nos.11301124 \& 61673145).

\end{document}